\documentclass[journal,draftclsnofoot,onecolumn,12pt,twoside]{IEEEtran}

\usepackage{makecell}
\usepackage[utf8]{inputenc} 
\usepackage[T1]{fontenc}
\usepackage{url}
\usepackage{ifthen}

\usepackage{times}
\usepackage[cmex10]{amsmath}
\usepackage{amssymb}
\usepackage{amsthm}
\usepackage{algorithm}
\usepackage{diagbox}
\usepackage{graphicx}
\usepackage{multirow}
\usepackage{tabularx}
\usepackage[bookmarks=false,colorlinks=false,pdfborder={0 0 0}]{hyperref}
\usepackage{cite}
\usepackage{bm}
\usepackage{arydshln}
\usepackage{mathtools}
\usepackage{microtype}
\usepackage{algorithmic}

\hypersetup{hidelinks}

\newtheorem{theorem}{Theorem}[section]
\newtheorem{lemma}[theorem]{Lemma}

\newtheorem{corollary}[theorem]{Corollary}

\long\def\symbolfootnote[#1]#2{\begingroup
\def\thefootnote{\fnsymbol{footnote}}\footnote[#1]{#2}\endgroup}
\renewcommand{\paragraph}[1]{{\bf #1}}
\title{Optimal Exponent of the Single-Error Correction Threshold with Fixed Redundancy for Analog Error-Correcting Codes}

\author{Zhengyi Jiang, Wenhao Liu, Zhongyi Huang, Hanxu Hou}

\begin{document}
\let\emph\textit
\maketitle
\symbolfootnote[0]{
Z. Jiang, W. Liu and Z. Huang are with the Department of Mathematics Sciences, Tsinghua University, Beijing, China~(E-mail: jzy10492761962@163.com,
wh-liu24@mails.tsinghua.edu.cn, zhongyih@tsinghua.edu.cn). 
H. Hou is with the Shenzhen University of Advanced Technology~(E-mail: houhanxu@163.com). \emph{(Corresponding author: Hanxu Hou.)}

This work was partially supported by the National Key R\&D Program of
China (No. 2025YFA1017200), the National Natural Science Foundation of China (No. 62371411, 61901115, 12025104).
}

\begin{abstract}
Analog error-correcting codes (Analog ECCs), introduced by Roth \cite{AECC2020}, address errors in vector-matrix multiplication arising from analog noise and sparse outliers in in-memory computing. 
A fundamental open problem concerns the lower bound on the single-error correction threshold $\Gamma_2(\mathcal C)$ for real $[n,k]$ linear codes with fixed redundancy $r=n-k\geq 2$.
Li \emph{et al.} \cite{AECC202606} recently established that for redundancy $r=2$, every real $[n,n-2]$ linear code $\mathcal{C}$ satisfies $\Gamma_2(\mathcal C)\geq \csc^2(\frac{\pi}{2n})$, resolving an open problem in \cite{AECC2020}, and showed that,
for every fixed $r \geq 2$, there exists a class of $[n,k]$ linear code $\mathcal{C}$ over $\mathbb{R}$ such that $\Gamma_2(\mathcal{C}) \leq  O(n^{1+\frac{1}{r-1}})$.

This paper proves the matching converse in \cite{AECC202606}. For every $[n,k]$ linear code $\mathcal{C}\subseteq \mathbb R^n$ with fixed redundancy $2\leq r<n$, we show that
\[
\Gamma_2(\mathcal C)\ge \frac{a_r}{\sqrt{r}\,\beta_{r-1}\,2^{\frac{1}{r-1}}}\cdot  n^{1+\frac{1}{r-1}},
\]
where $a_r =\frac{\Gamma(\frac{r}{2})}{\sqrt{\pi}\,\Gamma(\frac{r+1}{2})}$ and $\beta_d = \left(\frac{d\pi^{d-1}|\mathbb S^d|}{|\mathbb S^{d-1}|}\right)^{\frac{1}{d}}$ for positive integer $d$. 
Here $\mathbb S^d$ denotes the unit sphere in $\mathbb R^{d+1}$, $|\mathbb S^d|$ its surface area, and $\Gamma(\cdot)$ the Gamma function.
In particular, we further show that $\Gamma_2(\mathcal C)\geq \frac{1}{4\pi\sqrt{3}r}\cdot n^{1+\frac{1}{r-1}}$.
Together with the upper bound in \cite{AECC202606}, this confirms that the exponent $n^{1+\frac{1}{r-1}}$ is optimal, completing the asymptotic characterization of the single-error correction threshold for Analog ECCs.
\end{abstract}


\section{Introduction}
Real-valued vector-matrix multiplication underpins countless signal processing and machine learning operations \cite{Goodfellow-et-al-2016,hu2016dot,104196,sebastian2020memory,zhang2023edge,50305}, yet in practice its outputs are invariably corrupted by analog circuit noise, quantization errors, and channel disturbances.
To address this challenge, Roth introduced the framework of \emph{analog error-correcting codes} (Analog ECCs) ~\cite{AECC2019,AECC2020}. 
The Analog ECC model targets the corruption of real-valued vector outputs caused by small perturbations and large errors (outliers), with the objective of designing linear codes over $\mathbb{R}^n$ that can distinguish between these two error classes.

The rigorous mathematical model can be described as follows. Consider the vector-matrix multiplication $\boldsymbol{c} = \boldsymbol{u}A$, where $\boldsymbol{u} \in \mathbb{R}^{\ell}$ and $A \in \mathbb{R}^{\ell \times n}$. The computed result $\boldsymbol{y}$ contains small disturbances (denoted by $\boldsymbol{\varepsilon}=(\varepsilon_i)_{i\in[0,n-1]}$) and outlying errors (denoted by $\boldsymbol{e}=(e_i)_{i\in[0,n-1]}$), namely,
\begin{align}\nonumber
    \boldsymbol{y} = \boldsymbol{c} + \boldsymbol{\varepsilon} + \boldsymbol{e}.
\end{align}
It is usually assumed that the noise vector $\boldsymbol{\varepsilon}$ obeys $\lVert \boldsymbol{\varepsilon} \rVert_{\infty} \leq \delta$ for some prescribed small constant $\delta>0$. Large errors, defined by $|e_j| > \Delta$ for some $\Delta>\delta$ and $j\in[0,n-1]$, reside in the vector $\boldsymbol{e}$. In the Analog ECC model, the goal is todetect and locate the positions of large errors.
There has been extensive research in the literature on the theory of Analog ECCs \cite{AECC2024,AECC20242,2025Analog,zhu2026new,AECC202605,AECC202606}, the theoretical characterization and analysis of their error-correcting capability (or thresholds), as well as explicit code constructions and decoding algorithms.

The error-correcting capability of Analog ECCs is characterized by the threshold ($\frac{\Delta}{\delta}$). 
This paper focuses on the scenario of single-error correction. 
Roth in \cite{AECC2020} constructed a single-error correction MDS code (i.e., $n-k=2$), denoted as $\mathcal{C}_0$ in this paper, whose error-correction threshold is $$\Gamma_2(\mathcal{C}_0)=\csc^2(\frac{\pi}{2n}),$$ and posed the open problem \cite[Problem 4]{AECC2020}: ``Does the construction $\mathcal{C}_0$ have the smallest
possible value of $\Gamma_2(\cdot)$, among all $[n, n-2]$ linear codes
over $\mathbb{R}$?''.
Recently, Li \emph{et al.} \cite{AECC202606} resolved this problem by proving $$\Gamma_2(\mathcal{C})\geq \csc^2(\frac{\pi}{2n})$$ for any $[n, n-2]$ linear code $\mathcal{C}$. Furthermore, \cite{AECC202606} provided a specific $[n,n-r]$ linear code construction (denoted as $\mathcal{C}_1$) for $r\geq 2$ with $$\Gamma_2(\mathcal{C}_1) \leq 4n\left\lceil 
  \frac{n^{\frac{1}{r-1}} \Gamma\!\left(\frac{r+1}{2}\right)^{\frac{1}{r-1}}}{\sqrt{\pi}} 
  + \sqrt{\frac{r-1}{2}} 
\right\rceil= O(n^{1+\frac{1}{r-1}}).$$
They also derived an existential upper bound, showing that for sufficiently large $n$, there exists a real $[n,k]$ linear code $\mathcal{C}$ such that
\begin{align*}
    \Gamma_2(\mathcal{C})\leq \frac{\pi}{c_{r-1}}n^{1+\frac{1}{r-1}}, \ \text{where} \ c_d=\left(\frac{d|\mathbb S^d|}{|\mathbb S^{d-1}|}\right)^{\frac{1}{d}}.
\end{align*}
One natural question that arises is whether the leading term $n^{1+\frac{1}{r-1}}$ of the threshold $\Gamma_2(\cdot)$ is optimal for fixed $r\geq 3$.

This paper resolves the aforementioned gap by establishing the first universal lower bound on $\Gamma_2(\cdot)$ for fixed redundancy $r\geq 2$. 
\begin{enumerate}
    \item \textbf{Exact dual characterization of $\Gamma_2(\mathcal{C})$.} 
    We derive a new ordered-pair dual expression for the single-error correction threshold $\Gamma_2(\mathcal{C})$ of any real $[n,k]$ linear code with $n-k\geq 2$. 
    This expression (see Theorem~\ref{prop:dual}) reformulates the threshold as a constrained $\ell_1$-norm (which is defined in terms of the parity matrix) minimization problem over an affine set, providing a precise analytical handle for $\Gamma_2(\mathcal{C})$ that has not appeared in prior studies.

    \item \textbf{Normalization of the noise zonotope and parity matrix.} 
    We invoke the symmetric John's ellipsoid theorem to linearly normalize the noise zonotope $S_H$ (Eq. \eqref{EQ2}), ensuring that it is sandwiched between the Euclidean unit ball and its $\sqrt{r}$-multiple. 
    Based on this normalization, we introduce a averaging argument that yields a globally bound on the sum of the Euclidean norms of all columns of the parity matrix (see Eqs. \eqref{eqq2} to \eqref{eqq3}). 
    These novel techniques facilitate the analysis of the lower bound problem for $\Gamma_2(\mathcal{C})$ in the general setting of arbitrary redundancy $r\geq 2$.

    \item \textbf{Explicit uniform lower bound for all fixed redundancy.} 
    Building on the two preceding conclusions, we establish the first explicit and uniform lower bound 
    \[
    \Gamma_2(\mathcal{C}) \ge \frac{1}{4\pi\sqrt{3r}}\, n^{1+\frac{1}{r-1}}, \ 2\le r<n,
    \]
    for every real $[n,n-r]$ linear code $\mathcal{C}$. 
\end{enumerate}
Together with the upper bound obtained from the construction in \cite{AECC202606}, this implies that the leading term $n^{1+\frac{1}{r-1}}$ is optimal for fixed $r\geq 2$. 
Our results thus complete the theoretical framework for single-error correction threshold of Analog ECCs.

Although \cite{AECC2020} provided a qualitative characterization of the geometric features for $\Gamma_2(\mathcal{C})$ to some extent, it did not further offer a precise quantitative characterization; and Li \emph{et al.}~\cite{AECC202606} focused on a detailed trigonometric analysis for the special case $r=2$, however, this trigonometric analysis is only valid for the case $r=2$ and cannot be extended to higher-dimensional settings where $r\geq 3$.
Jiang \emph{et al.} in \cite{AECC20242} established an explicit optimization model for the value of $h_m(\mathcal{C})$ for general $m \ge 1$, and showed that solving it is equivalent to solving $n(n-1)\binom{n-1}{m-1}2^m$ linear programming (LP) subproblems with $k$ real-valued variables. 
Roth \emph{et al.} in \cite{roth2026height} recently further optimized the algorithms in \cite{AECC20242} for computing $h_m(\mathcal{C})$. 
However, these optimization problems are only applicable to the approximate evaluation of $h_m(\mathcal{C})$ for a given specific linear code $\mathcal{C}$ (and since the computational complexity of solving the sub-LP problems grows exponentially, it is typically only applicable to small code lengths $n$), and it is difficult to use them to obtain a deeper characterization of the bounds in the general case.
Despite prior efforts, the problem of providing an exact theoretical characterization of $\Gamma_2(\mathcal{C})$ for every fixed $r\ge 3$ remains unresolved.


The remainder of this paper is organized as follows. 
Section~\ref{sec:prelim} reviews prior work and necessary mathematical preliminaries.
Section~\ref{sec:duality} establishes an exact ordered-pair dual expression for $\Gamma_2(\mathcal{C})$. 
Section~\ref{sec:geometry} analyzes several key geometric properties and the corresponding normalization results.
Section~\ref{sec:lower} proves the main lower bound and derives explicit constants.
Section~\ref{sec:con} concludes the paper.

\section{Preliminaries}\label{sec:prelim}

\subsection{Analog Error-Correcting Codes}
In this section, we review the relevant definitions of Analog ECCs and the key results in related works.

\textbf{$(\tau,\sigma)$ decoder for ($\mathcal{C}, \Delta:\delta$).} For an $[n=k+r, k]$ linear code $\mathcal{C}$ over $\mathbb{R}$, a decoder for $\mathcal{C}$ is a function $\mathcal{D}: \mathbb{R}^n \to 2^{[0,n-1]} \cup \{\text{"e"}\}$ that returns a set of locations of outlying errors or an indication "e" that errors have been detected.
For $\Delta\geq 0$, and any vector $\boldsymbol{e}=(e_i)_{i\in[0,n-1]}\in\mathbb{R}^n$, 
$\operatorname{Supp}_{\Delta}(\boldsymbol{e}):=\{i\in[0,n-1]||e_i|>\Delta\}$.
Specifically, a $(\tau,\sigma)$ decoder for ($\mathcal{C}, \Delta:\delta$) is defined as \cite{AECC2020}: Given $\delta, \Delta \in \mathbb{R}^+$ and nonnegative integers $\tau$ and $\sigma$, a decoder $\mathcal{D}$ that can correct $\tau$ errors and detect $\sigma$ additional errors with the threshold pair $(\delta, \Delta)$  is equivalent to satisfying the following two conditions ((D1) and (D2)) for every $\boldsymbol{y}=\boldsymbol{c}+\boldsymbol{\varepsilon}+\boldsymbol{e}$, where $\boldsymbol{c} \in \mathcal{C}$, $\boldsymbol{\varepsilon} \in \mathcal{Q}(n, \delta)$, and $\boldsymbol{e} \in \mathcal{B}(n, \tau + \sigma)$. Here, set $\mathcal{Q}(n,\delta):=\{\boldsymbol{\varepsilon}=(\varepsilon_i)_{i\in[0,n-1]}||\varepsilon_i|\leq \delta, \ \forall i\in[0,n-1]\}$, and $\mathcal{B}(n,m)$ represents the set of vectors in $\mathbb{R}^{n}$ whose Hamming weight does not exceed $m$.

(D1) If $\boldsymbol{e} \in \mathcal{B}(n, \tau)$ then "e" $\neq \mathcal{D}(\boldsymbol{y})\subseteq \operatorname{Supp}_{0}(\boldsymbol{e})$;

(D2) If $\mathcal{D}(\boldsymbol{y}) \neq \text{``e''}$ then $\operatorname{Supp}_{\Delta}(\boldsymbol{e}) \subseteq \mathcal{D}(\boldsymbol{y})$.

\textbf{Threshold Results.} The error-correction threshold directly reflects the error-correcting capability, and its lower bound serves as an important criterion for the existence of a decoder.
The threshold is characterized by the \emph{height}.
For a nonzero $\boldsymbol{c}=(c_i)_{i\in[0,n-1]}\in\mathbb R^n$, write its ordered coordinate magnitudes as ($\{|c_i|\}_{i=0}^{n-1}=\{|c|_{(j)}\}_{j=0}^{n-1}$)
\[
|c|_{(0)}\geq |c|_{(1)}\ge |c|_{(2)}\ge \cdots \ge |c|_{(n-1)}.
\]
Define the $m$-height
\[
h_m(\boldsymbol{c})=\frac{|c|_{(0)}}{|c|_{(m)}},
\]
with value $+\infty$ when $|c|_{(m)}=0$. For an $[n,k]$ linear code $\mathcal{C}$ over $\mathbb{R}$, define its $m$-height as 
\[
h_2(\mathcal C)=\sup_{\boldsymbol{c}\in\mathcal C\setminus\{\boldsymbol{0}\}} h_2(\boldsymbol{c}).
\]
The following theorem gives the necessary and sufficient
conditions for the existence of the $(\tau,\sigma)$ decoder for ($\mathcal{C}, \Delta:\delta$).
\begin{theorem}\label{TH.01}\cite[Theorem 1]{AECC2020}There exists a $(\tau,\sigma)$ decoder for ($\mathcal{C}, \Delta:\delta)$ if and only if\begin{align*}    \frac{\Delta}{\delta}\geq 2h_{2\tau+\sigma}(\mathcal{C})+2.\end{align*}\end{theorem}
For $m\in[0,n+1]$, define $\Gamma_m(\mathcal{C})=2h_{m}(\mathcal{C})+2$.
According to theorem \ref{TH.01}, $\Gamma_{2\tau +\sigma}(\mathcal{C})$  serves as the lower bound of the error-correction threshold of code $\mathcal{C}$; however, its exact value as expressed is difficult to characterize.
A key challenge is to determine 
\begin{align}\nonumber
\Gamma_m(n,k):=\inf\{\Gamma_m(\mathcal C): \mathcal C\subseteq\mathbb R^n \text{ is a linear } [n,k] \text{ code}\}.
\end{align}

\textbf{Related works.}
Roth's work \cite{AECC2019,AECC2020} introduced the framework of Analog ECCs and focused on constructions for single-error detection and single-error correction $[n,k]$ linear codes. 
Wei \emph{et al.} in \cite{AECC2024} proposed constructions for multiple-error-correcting Analog ECCs. 
The works in \cite{AECC20242,roth2026height} developed algorithms based on LP problems for computing general $m$-heights $h_m(\mathcal{C})$. 
Jiang \emph{et al.} \cite{2025Analog} provided explicit constructions of single-error-correcting codes with $\Gamma_2(\mathcal{C})=O(n^2)$ and  efficient decoding algorithms. 
Zhu \emph{et al.} \cite{zhu2026new} proposed a geometric construction that handles multiple outliers and analyzed its height profiles.
Song \emph{et al.} \cite{song2026AECC} proposed constructions of single-error correction codes with redundancy $r=3$, whose error-correction threshold attains $\Gamma_2(\mathcal{C})=O(n\sqrt{n})$.
Recently, Jiang \emph{et al.} \cite{AECC202605} derived tight lower bound for the single-error detection threshold $\Gamma_1(\mathcal C)\geq \frac{2n}{n-k}$ under the condition $2 \le n-k \mid k$. 
Li \emph{et al.} \cite{AECC202606} extended the results of \cite{AECC202605} to the case where $n-k$ does not divide $k$, establishing lower bound for $\Gamma_1(\mathcal{C})\geq 2\left\lceil \frac{n}{n-k} \right\rceil$, 
and they proved the tight lower bound $\Gamma_2(\mathcal C)\ge \csc^2(\frac{\pi}{2n})$ for $r=2$ and constructed codes with $\Gamma_2(\mathcal{C}) = O(n^{1+\frac{1}{r-1}})$ for fixed $r\ge2$.

Following the threshold bound for single-error-correcting Analog ECCs, this work proves the lower bound for $\Gamma_2(\mathcal{C})$ at arbitrary fixed redundancy $r\geq 2$, thereby confirming the optimality of the exponent $n^{1+\frac{1}{r-1}}$.

\subsection{Mathematical Preliminaries}

In this section, we first introduce in Table~\ref{tab:notation} the mathematical notation commonly used in this paper.
We next present the main convex optimization and geometric tools employed in this paper: the strict separation theorem, the symmetric John's ellipsoid theorem, and the support function.
\begin{table}[htbp]
\centering
\caption{Summary of Notation.}
\label{tab:notation}
\renewcommand{\arraystretch}{1.15}
\begin{tabularx}{\textwidth}{>{\centering\arraybackslash}p{0.20\textwidth}|>{\raggedright\arraybackslash}X}
\hline
\hline
Notation & Description \\
\hline
$[a,b]$& For non-negative integers $a$ and $b$ with $a<b$, $[a,b]:=\{a,a+1,\ldots,b\}$.\\
\hline
$\mathbb N$&The set of nonnegative integers.\\
\hline
$\mathbb R$ & Field of real numbers. \\
\hline
$\mathbb S^{n}$ & Unit sphere in $\mathbb R^{n+1}$. \\
\hline
$I_n$&The identity matrix on $\mathbb{R}^{n \times n}$.\\
\hline
$A^T$&The transpose of matrix $A$.\\
\hline
$rank(A)$&The rank of $A\in\mathbb{R}^{m\times n}$.\\
\hline
$\det(A)$& The determinant of matrix $A$.\\
\hline
$\operatorname{Ker} (A)$& The kernel space of matrix $A\in\mathbb R^{m\times n}$, i.e.,
$\operatorname{Ker} (A) = \left\{ \boldsymbol{x} \in \mathbb{R}^n \mid A \boldsymbol{x} = \boldsymbol{0} \right\}.$\\
\hline
$\|\boldsymbol{x}\|_1$, $\|\boldsymbol{x}\|_2$, $\|\boldsymbol{x}\|_\infty$ & $\ell_1$-norm, $\ell_2$-norm, and $\ell_\infty$-norm of $\boldsymbol{x}\in\mathbb R^n$ \\
\hline
\(\boldsymbol{e}_i\)&The unit vector whose $i$-th component is 1 and all others are 0.\\
\hline
$[\boldsymbol{v}]$& For nonzero $\boldsymbol{v} \in \mathbb R^n$, define $
[\boldsymbol{v}] = \operatorname{span}\{\boldsymbol{v}\} = \{\lambda \boldsymbol{v} \mid \lambda \in \mathbb{R}\}$. This is the line through the origin generated by $\boldsymbol{v}$.\\
\hline
$\langle \boldsymbol{x}, \boldsymbol{y} \rangle$& The inner product of $\boldsymbol{x}$ and $\boldsymbol{y}$, where $\boldsymbol{x},\boldsymbol{y}\in\mathbb{R}^n$.\\
			\hline
$B_2^n$ & The unit ball in $\mathbb{R}^n$, i.e., $B_2^n=\left\{ \boldsymbol{x} \in \mathbb{R}^n | \|\boldsymbol{x}\|_2 \le 1 \right\}$.\\
\hline
$\Gamma(x)$ &The Gamma function, i.e., $\Gamma(x)=\int_{0}^{\infty} t^{x-1} e^{-t} dt, \  x > 0.$\\
\hline
$\mathbb E[X]$& The expectation of the random variable $X$.\\
\hline
Euclidean ball&For a center $\boldsymbol{x}_0$ $\in \mathbb{R}^n$ and radius $r > 0$, the (closed) Euclidean ball is defined as $\left\{ \boldsymbol{x} \in \mathbb{R}^n : \|\boldsymbol{x}-\boldsymbol{x}_0\|_2 \le r \right\}$.\\
\hline
Ellipsoid &\makecell{ A (closed) ellipsoid in $\mathbb{R}^n$ is defined via a symmetric positive definite matrix $Q$ and a center \\$\boldsymbol{x}_0 \in \mathbb{R}^n$ as
$\mathcal{E} = \left\{ \boldsymbol{x} \in \mathbb{R}^n \mid (\boldsymbol{x} - \boldsymbol{x}_0)^T Q^{-1} (\boldsymbol{x} - \boldsymbol{x}_0) \le 1 \right\}.$}\\
\hline
\hline
\end{tabularx}
\end{table}

\begin{theorem}[Strict separation of a point and a closed convex set{\cite{boyd2004convex}}]\label{TH.2.2}	Let $\Omega \subseteq \mathbb{R}^n$ be a nonempty, closed, convex set, and let $\boldsymbol{y} \in \mathbb{R}^n$ be a point such that $\boldsymbol{y} \notin \Omega$. Then there exists a nonzero vector $\boldsymbol{u} \in \mathbb{R}^n$ and a scalar $c \in \mathbb{R}$ such that		$$	\langle \boldsymbol{u}, \boldsymbol{y} \rangle > c \quad \text{and} \quad \langle \boldsymbol{u}, \boldsymbol{x} \rangle \leq c \quad \text{for all } \boldsymbol{x} \in \Omega.	$$		Equivalently, there exists a hyperplane $H = \{\boldsymbol{z} \in \mathbb{R}^n : \langle \boldsymbol{u}, \boldsymbol{z} \rangle = c\}$ that strictly separates $\boldsymbol{y}$ from $\Omega$ (i.e., $\boldsymbol{y}$ lies in the open half-space and $\Omega$ lies in the closed half-space determined by $H$).\end{theorem}

\begin{theorem}[Symmetric John's ellipsoid theorem {\cite{ball1997}}]\label{thm:symmetric-john}
Let $K\subseteq\mathbb R^r$ be a centrally symmetric convex body, meaning that $K$ is compact and convex, has nonempty interior in $\mathbb R^r$, and satisfies $K=-K$. If $E$ is the ellipsoid of maximum volume contained in $K$, then $E$ is centered at the origin and
\[
E\subseteq K\subseteq \sqrt r\,E.
\]
\end{theorem}

Next, we introduce an important tool, the \emph{support function}.
Let \( S \subseteq \mathbb{R}^n \) be a nonempty set. Its support function \( h_S : \mathbb{R}^n \to \mathbb{R} \cup \{+\infty\} \) is defined by
\[
h_S(\boldsymbol{y}) = \sup_{\boldsymbol{x} \in S} \, \langle \boldsymbol{y}, \boldsymbol{x} \rangle.
\]
When $S$ is a closed convex set, the support function serves as an equivalent characterization of set inclusion. Specifically, we prove the following result.

\begin{corollary}\label{TH.2.4}
    Let $S \subseteq \mathbb{R}^n$ be a nonempty closed convex set, and let $A \subseteq \mathbb{R}^n$ be nonempty. Then
\[
A \subseteq S \quad \Longleftrightarrow \quad h_A(\boldsymbol{y}) \le h_S(\boldsymbol{y}),\ \text{for all } \boldsymbol{y} \in \mathbb{R}^n.
\]
\end{corollary}
\begin{proof}
    We prove each direction separately.
    
($\Rightarrow$) Assume $A \subseteq S$. For any fixed $\boldsymbol{y} \in \mathbb{R}^n$, since every $\boldsymbol{x} \in A$ also belongs to $S$, we have
\[
\langle \boldsymbol{y}, \boldsymbol{x} \rangle \le \sup_{\boldsymbol{s} \in S} \langle \boldsymbol{y}, \boldsymbol{s} \rangle = h_S(\boldsymbol{y}).
\]
Taking the supremum over all $\boldsymbol{x} \in A$ yields
\[
h_A(\boldsymbol{y}) = \sup_{\boldsymbol{x} \in A} \langle \boldsymbol{y}, \boldsymbol{x} \rangle \le h_S(\boldsymbol{y}),\ \forall \boldsymbol{y}\in\mathbb R^n.
\]

($\Leftarrow$) Conversely, suppose that
\[
h_A(\boldsymbol{y}) \le h_S(\boldsymbol{y}), \   \text{for all } \boldsymbol{y} \in \mathbb{R}^n.
\]
We claim that $A \subseteq S$. Take any $\boldsymbol{x} \in A$ and assume, for contradiction, that $\boldsymbol{x} \notin S$.

Since $S$ is closed and convex, Theorem \ref{TH.2.2} guaranties the existence of a vector $\boldsymbol{y} \in \mathbb{R}^n$ such that
\[
\langle \boldsymbol{y}, \boldsymbol{x} \rangle > \sup_{\boldsymbol{s} \in S} \langle \boldsymbol{y}, \boldsymbol{s} \rangle = h_S(\boldsymbol{y}).
\]
Then we have,
\[
h_A(\boldsymbol{y}) \ge \langle \boldsymbol{y}, \boldsymbol{x} \rangle > h_S(\boldsymbol{y}),
\]
which contradicts the assumption that $h_A(\boldsymbol{y}) \le h_S(\boldsymbol{y})$ for all $\boldsymbol{y}\in\mathbb R^n$.
Therefore, the proof is complete.
\end{proof}

\textbf{Remark.}
Note that, for any $\boldsymbol{y} \neq \boldsymbol{0}$, $h_A(\boldsymbol{y}) \le h_S(\boldsymbol{y})$ is equivalent to $h_A\left(\boldsymbol{\frac{\boldsymbol{y}}{\|\boldsymbol{y}\|_2}}\right) \le h_S\left(\boldsymbol{\frac{\boldsymbol{y}}{\|\boldsymbol{y}\|_2}}\right)$; therefore, the statement of Corollary \ref{TH.2.4} can be reduced to
\begin{align*}
    A \subseteq S \quad \Longleftrightarrow \quad h_A(\boldsymbol{y}) \le h_S(\boldsymbol{y}),\  \text{for all } \boldsymbol{y} \in \mathbb{S}^{n-1}.
\end{align*}

\section{An Exact Ordered-Pair Dual Expression for $\Gamma_2(\mathcal{C})$}\label{sec:duality}

\subsection{Preliminaries on the Parity Matrix}
Let $H=(\boldsymbol{h}_0,\boldsymbol{h}_1,\ldots,\boldsymbol{h}_{n-1})\in\mathbb R^{r\times n}$ be a parity matrix of an $[n,k=n-r]$ linear code $\mathcal C$ with $rank(H)=r$, so that $\mathcal C=\operatorname{Ker} (H)$. Define the centrally symmetric zonotope
\begin{align}\label{EQ2}
    S_H=\left\{\sum_{i=0}^{n-1} x_i \boldsymbol{h}_i\mid |x_i|\le 1, \ \forall i\in[0,n-1]\right\}.
\end{align}
$S_H$ describes the region in the syndrome space that is affected by noise.
The support function of $S_H$ is
\begin{align}\label{eq1}
h_{S_H}(\boldsymbol{u})&=\sup_{\boldsymbol{z}\in S_H}\langle \boldsymbol{u},\boldsymbol{z}\rangle\nonumber\\
&=\sup_{|x_i|\le 1, \ \forall i\in[0,n-1]}\langle \boldsymbol{u},\sum_{i=0}^{n-1} x_i \boldsymbol{h}_i\rangle\nonumber\\
&=\sum_{k=0}^{n-1} |\langle \boldsymbol{u},\boldsymbol{h}_k\rangle|=\|H^T\boldsymbol{u}\|_1. 
\end{align}
Since $H$ has full row rank (which means $\operatorname{Ker} (H^T)=\boldsymbol{0}$), Eq. \eqref{eq1} defines a norm on $\mathbb R^r$.
In the subsequent discussion, we will assume that the parity matrix $H$ satisfies certain non-degeneracy conditions. The justification for this assumption is given in the following lemma.

\begin{lemma}\label{lem:degenerate}
If one column of $H$ is $\boldsymbol{0}$ or if two columns of $H$ are proportional, then $\Gamma_2(\mathcal C)=+\infty$. Otherwise, every column is nonzero and the $n$ projective directions  $\{[\boldsymbol{h}_i]\}_{i=0}^{n-1}$ are distinct.
\end{lemma}

\begin{proof}
If $\boldsymbol{h}_i=\boldsymbol{0}$ for certain $i\in[0,n-1]$, then the $i$-th standard basis vector belongs to $\operatorname{Ker} H$. If $\boldsymbol{h}_i=\lambda \boldsymbol{h}_t$ for $i\neq t$ and $\lambda\in\mathbb R$, then $\boldsymbol{e}_i-\lambda \boldsymbol{e}_t\in\operatorname{Ker} H$. In either case, $\mathcal C$ contains a nonzero codeword of Hamming weight at most two. Its third-largest coordinate is zero, so $h_2(\mathcal C)=+\infty$ and hence $\Gamma_2(\mathcal C)=+\infty$. Conversely, if all columns are nonzero and pairwise non-proportional, then $\operatorname{Ker} H$ contains no nonzero codeword supported on one or two coordinates. Hence every nonzero codeword has a strictly positive third-largest coordinate. Since the unit sphere in $\mathcal C$ is compact, this also implies that  $\Gamma_2(\mathcal{C})<+\infty$.
\end{proof}

By Lemma \ref{lem:degenerate}, the desired finite lower bound is automatic in the degenerate case. We henceforth assume that all columns of $H$ are nonzero and pairwise non-proportional.

\subsection{Dual Expression for $\Gamma_2(\mathcal{C})$}
In this section, we derive an exact ordered-pair dual expression for $\Gamma_2(\mathcal{C})$.
First, we prove an auxiliary lemma.

\begin{lemma}\label{lem:boxduality}
Let $B\in\mathbb R^{r\times m}$ and let $\boldsymbol{p},\boldsymbol{q}\in\mathbb R^r$ satisfy $\boldsymbol{p}\notin\operatorname{span}(\boldsymbol{q})$. Set
\[
P=\sup\{a\ge 0\mid \text{there exist }\ \boldsymbol{x}\in\mathbb R^m \ \text{with}\  \|\boldsymbol{x}\|_{\infty}\leq 1\ \text{and}\ \ b\in\mathbb R\ \text{ such that } B\boldsymbol{x}+a\boldsymbol{p}+b\boldsymbol{q}=\boldsymbol{0}\}.
\]
Then
\begin{align}\label{eq2}
P=\inf_{\substack{\langle \boldsymbol{u},\boldsymbol{q}\rangle=0\\ \langle \boldsymbol{u},\boldsymbol{p}\rangle=1}}\|B^T\boldsymbol{u}\|_1. 
\end{align}
\end{lemma}

\begin{proof}
Since $\boldsymbol{p}\notin\operatorname{span}(\boldsymbol{q})$, the feasible set on the right-hand side of Eq. \eqref{eq2} is nonempty, i.e., the constraint  $\langle \boldsymbol{u},\boldsymbol{q}\rangle=0$ and $\langle \boldsymbol{u},\boldsymbol{p}\rangle=1$ is feasible.
For every primal feasible triple $(\boldsymbol{x},a,b)$ and every dual feasible $\boldsymbol{u}$, we have
\begin{align*}
    0=\langle \boldsymbol{u},B\boldsymbol{x}+a\boldsymbol{p}+b\boldsymbol{q}\rangle=\langle \boldsymbol{u},B\boldsymbol{x}\rangle+a\langle \boldsymbol{u},\boldsymbol{p}\rangle+b\langle \boldsymbol{u},\boldsymbol{q}\rangle
    =\langle \boldsymbol{u},B\boldsymbol{x}\rangle+a.
\end{align*}
Then we have 
\begin{align*}
    a=-\langle \boldsymbol{u},B\boldsymbol{x}\rangle\le \|B^T\boldsymbol{u}\|_1\|\boldsymbol{x}\|_\infty\le \|B^T\boldsymbol{u}\|_1.
\end{align*}
This proves the $\le$ direction in Eq. \eqref{eq2}.

For the reverse direction, let $$K = \left\{ B\boldsymbol{x} + \ell  \boldsymbol{q} \;\middle|\; \|\boldsymbol{x}\|_{\infty}\leq 1,\; \ell \in \mathbb{R} \right\}.$$ 
This set is convex and centrally symmetric. It is also closed. Let $C_m=\{\boldsymbol{x}\in\mathbb R^m\mid \|\boldsymbol{x}\|_\infty\le 1\}$, and write $BC_m=\{B\boldsymbol{x}\mid \boldsymbol{x}\in C_m\}$. Since $C_m$ is compact and $B$ is linear, $BC_m$ is compact. Let $\pi:\mathbb R^r\to \mathbb R^r/\operatorname{span}(\boldsymbol{q})$ be the quotient map. Since $\pi$ is continuous, $\pi(BC_m)$ is compact and hence closed in the finite-dimensional quotient space. Moreover, $K=\pi^{-1}(\pi(BC_m))$, so $K$ is closed. Write the right-hand side of Eq. \eqref{eq2} as $\lambda$. Suppose that $0\le t<\lambda$ but $t\boldsymbol{p}\notin K$. According to Theorem \ref{TH.2.2}, a point outside a closed convex set can be strictly separated from it. Hence, for some $\boldsymbol{v}\neq \boldsymbol{0}$, we have
\begin{align*}
    t\langle \boldsymbol{v},\boldsymbol{p}\rangle>\sup_{\boldsymbol{z}\in K}\langle \boldsymbol{v},\boldsymbol{z}\rangle.
\end{align*}
Note that, since
\begin{align}\label{eq3}
    \sup_{\boldsymbol{z}\in K}\langle \boldsymbol{v},\boldsymbol{z}\rangle=\sup_{\|\boldsymbol{x}\|_{\infty}\leq 1,\; \ell \in \mathbb{R}}\langle \boldsymbol{v},B\boldsymbol{x} + \ell  \boldsymbol{q}\rangle=\sup_{\|\boldsymbol{x}\|_{\infty}\leq 1,\; \ell \in \mathbb{R}}(\langle \boldsymbol{v},B\boldsymbol{x}\rangle + \ell  \langle\boldsymbol{v},\boldsymbol{q}\rangle)
\end{align}
and $\ell \in \mathbb R$, the supremum on the right-hand side of Eq. \eqref{eq3} can be finite only if $\langle \boldsymbol{v},\boldsymbol{q}\rangle=0$, 
in which case we have $\sup_{\boldsymbol{z}\in K}\langle \boldsymbol{v},\boldsymbol{z}\rangle=\|B^T\boldsymbol{v}\|_1$. 
Since $\|B^T\boldsymbol{v}\|_1$ is nonnegative, $\langle \boldsymbol{v},\boldsymbol{p}\rangle>0$. 
Therefore, $\boldsymbol{u}=\frac{\boldsymbol{v}}{\langle \boldsymbol{v},\boldsymbol{p}\rangle}$ is dual feasible and satisfies 
$$\|B^T\boldsymbol{u}\|_1=\frac{\|B^T\boldsymbol{v}\|_1}{\langle \boldsymbol{v},\boldsymbol{p}\rangle}<t$$ 
according to Eq. \eqref{eq3}, contradicting $t<\lambda$. 
Thus $t\boldsymbol{p}\in K$ for every $0\le t<\lambda$. Because $K=-K$, we also have $-t\boldsymbol{p}\in K$. 
Thus the original problem is feasible with $a=t$, and $P\ge \lambda$.
The proof is complete.
\end{proof}

Next, we establish the dual expression for $\Gamma_2(\mathcal{C})$.
For any ordered pair $(i,j)$ with $i,j\in [0,n-1]$ and $i\ne j$, define
\begin{align}\label{eq3.5}
    P_{ij}=\sup\{c_i \mid \boldsymbol{c}\in\mathcal C,\ |c_k|\le 1\text{ for every }k\in[0,n-1]\setminus\{i,j\}\}.
\end{align}
Since $\mathcal C$ is symmetric, the same value is obtained with $|c_i|$ in the objective.

\begin{theorem}\label{prop:dual}
For every parity matrix $H=(\boldsymbol{h}_0,\ldots,\boldsymbol{h}_{n-1})$ with full row rank whose columns are nonzero and pairwise non-proportional, we have
\begin{align}\label{eq3.6}
h_2(\mathcal C)=\max_{i\ne j} P_{ij}
\end{align}
and
\begin{align}\label{eq4}
P_{ij}=\min_{\substack{\langle \boldsymbol{u},\boldsymbol{h}_j\rangle=0\\ \langle \boldsymbol{u},\boldsymbol{h}_i\rangle=1}}\sum_{t\notin\{i,j\}}|\langle \boldsymbol{u},\boldsymbol{h}_t\rangle|.
\end{align}
Consequently,
\begin{align}\label{eq5}
\Gamma_2(\mathcal C)=2\max_{i\ne j}\min_{\substack{\langle \boldsymbol{u},\boldsymbol{h}_j\rangle=0\\ \langle \boldsymbol{u},\boldsymbol{h}_i\rangle=1}}\sum_{t=0}^{n-1} |\langle \boldsymbol{u},\boldsymbol{h}_t\rangle|.
\end{align}
Furthermore, the minimum of each of Eqs. \eqref{eq4} and \eqref{eq5} is attained.
\end{theorem}

\begin{proof}
If a nonzero $\boldsymbol{c}$ is feasible in the definition of $P_{ij}$ in Eq. \eqref{eq3.5}, then at most two coordinates are unconstrained, so $|c|_{(2)}\le 1$ (otherwise, if $|c|_{(2)}> 1$, then $|c|_{(0)}\geq |c|_{(1)}\geq |c|_{(2)}>1$, which contradicts that $\boldsymbol{c}$ is feasible). Hence
\begin{align*}
    |c_i|\le |c|_{(0)}\le h_2(\boldsymbol{c})=\frac{|c|_{(0)}}{|c|_{(2)}}\le h_2(\mathcal C). 
\end{align*}
This proves $P_{ij}\le h_2(\mathcal C)$.

In contrast, take any nonzero $\boldsymbol{c}\in\mathcal C$. Choose $i$ so that $|c_i|=|c|_{(0)}$, and choose $j\ne i$ so that $|c_j|$ is the second-largest coordinate magnitude. Divide $\boldsymbol{c}$ by $|c|_{(2)}$. After changing the overall sign if necessary, the resulting codeword is feasible for $P_{ij}$, and the $i$-th coordinate equals $h_2(\boldsymbol{c})$.
Then we have $c_i=h_2(\boldsymbol{c})\leq P_{ij}$.
Taking the supremum over $\boldsymbol{c}$ proves that $$h_2(\mathcal C)=\sup_{\boldsymbol{c}\in\mathcal C\setminus\{\boldsymbol{0}\}} h_2(\boldsymbol{c})\leq \max_{i\ne j} P_{ij}.$$ 
Therefore, Eq. \eqref{eq3.6} is proved.

Fix $i\ne j$, let $R:=[0,n-1]\setminus\{i,j\}$, and write $H_R=(\boldsymbol{h}_t)_{t\in R}$. 
By the symmetry noted above, we may restrict to the case $a=c_i\geq 0$. Such a value is feasible for $P_{ij}$ exactly when there are $\boldsymbol{x}\in\mathbb{R}^{n-2}$ with $\|\boldsymbol{x}\|_{\infty}\leq 1$ and $b=c_j\in\mathbb R$ such that
\[
H_R \boldsymbol{x}+a \boldsymbol{h}_i+b \boldsymbol{h}_j=\boldsymbol{0}.
\]
Applying Lemma~\ref{lem:boxduality} with $B=H_R$, $\boldsymbol{p}=\boldsymbol{h}_i$, and $\boldsymbol{q}=\boldsymbol{h}_j$ gives 
\begin{align}\label{eq8}
    P_{ij}=\inf_{\substack{\langle \boldsymbol{u},\boldsymbol{h}_j\rangle=0\\ \langle \boldsymbol{u},\boldsymbol{h}_i\rangle=1}}\|H_R^T\boldsymbol{u}\|_1=\inf_{\substack{\langle \boldsymbol{u},\boldsymbol{h}_j\rangle=0\\ \langle \boldsymbol{u},\boldsymbol{h}_i\rangle=1}}\sum_{t\notin\{i,j\}}|\langle \boldsymbol{u},\boldsymbol{h}_t\rangle|. 
\end{align}

Note that, for any $\boldsymbol{u}$ in the feasible affine set in $P_{ij}$,
\[
\|H^T\boldsymbol{u}\|_1=\sum_{t=0}^{n-1} |\langle \boldsymbol{u},\boldsymbol{h}_t\rangle|=1+\sum_{t\notin\{i,j\}}|\langle \boldsymbol{u},\boldsymbol{h}_t\rangle|.
\]
Since $H$ has full row rank, $\|H^T\boldsymbol{u}\|_1$ is a norm on $\mathbb R^r$. Its restriction to the closed feasible affine set has compact sublevel sets, so the minimum is attained. The minimum in Eq. \eqref{eq8} is attained as well because the excluded two terms are fixed on this affine set.
Therefore, Eq. \eqref{eq4} is proved.

Finally, we have
\[
\Gamma_2(\mathcal C)=2h_2(\mathcal C)+2
=2\max_{i\ne j}(P_{ij}+1)
=2\max_{i\ne j}\min_{\substack{\langle \boldsymbol{u},\boldsymbol{h}_j\rangle=0\\ \langle \boldsymbol{u},\boldsymbol{h}_i\rangle=1}}\sum_{t=0}^{n-1} |\langle \boldsymbol{u},\boldsymbol{h}_t\rangle|.
\]
Therefore, Eq. \eqref{eq5} is proved.
\end{proof}

\section{Several Key Geometric Properties}\label{sec:geometry}

Before deriving the lower bound for $\Gamma_2(\mathcal{C})$, we first prove two important geometric results in this section. 
We first derive the normalization lemma from Theorem \ref{thm:symmetric-john} stated in Section~\ref{sec:prelim}.

\begin{lemma}\label{lem:john}
Let $K\subseteq\mathbb R^r$ be a centrally symmetric convex body. Then there exists an invertible linear map $A$ such that
\begin{equation}
B_2^r\subseteq AK\subseteq \sqrt r\,B_2^r.
\end{equation}
\end{lemma}

\begin{proof}

By Theorem~\ref{thm:symmetric-john}, there is an origin-centered ellipsoid $E$ such that
\[
E\subseteq K\subseteq \sqrt r\,E.
\]
Since $E$ has nonempty interior, there is an invertible linear map $A$ that sends $E$ to $B_2^r$. Applying $A$ to the inclusion gives
\[
B_2^r\subseteq AK\subseteq \sqrt r\,B_2^r.
\]

\end{proof}

The significance of Lemma~\ref{lem:john} lies in that it allows us to normalize the parity matrix and the noise part of the syndrome vector for any $[n,k]$ linear code $\mathcal{C}$. 
Specifically, for parity matrix $H$ of the code $\mathcal{C}$, applying Lemma~\ref{lem:john} with $K=S_H$, there exists an invertible matrix $A \in \mathbb{R}^{r \times r}$ such that $\mathcal{C} =\operatorname{Ker}(H) = \operatorname{Ker}(AH)$ (denote $G = AH$), and we have
\begin{align}\label{eqAAA}
    B_2^r \subseteq S_G = A S_H \subseteq \sqrt{r}B_2^r.
\end{align}
Henceforth, without loss of generality, we take $G$ as the parity matrix for code $\mathcal{C}$.
Fig. \ref{fig:1} illustrates the geometric intuition behind the set transformation (shown in two-dimensional case).

\begin{figure}[htpb]
		\centering
		\includegraphics[width=0.9\linewidth]{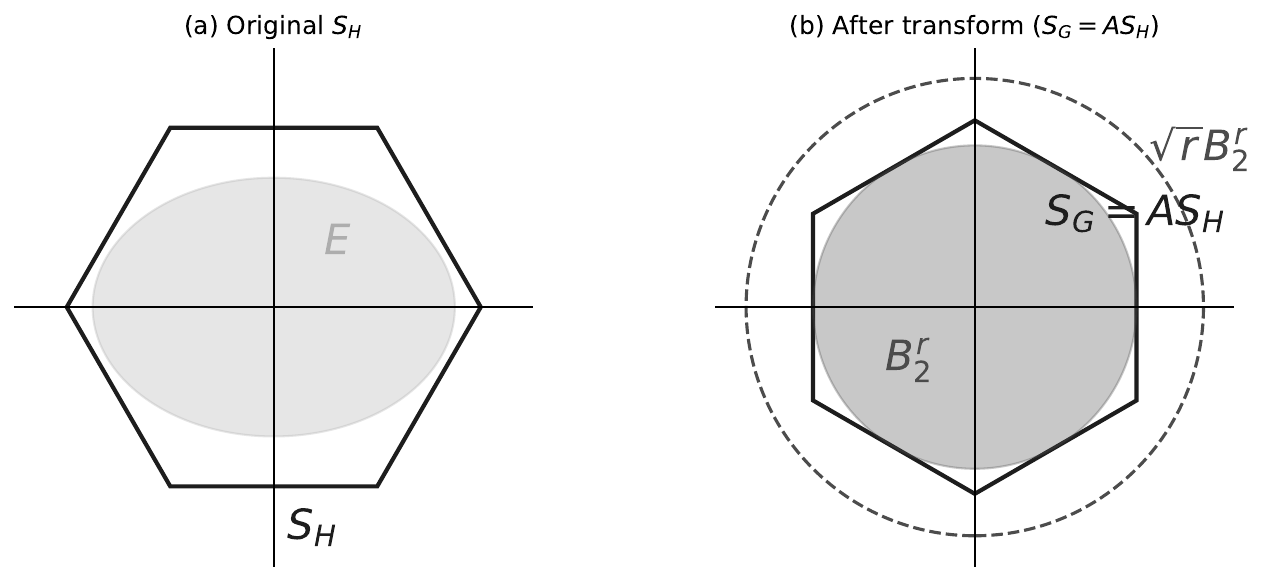}
		\caption{Illustration of the normalization of $S_H$ (in the two-dimensional case, i.e., $r=2$), where $E$ is the ellipsoid of maximal volume inscribed in $S_H$.}
		\label{fig:1}
	\end{figure}

Finally, we establish another auxiliary lemma in this section.
For a nonzero vector $\boldsymbol{v}$, write $[\boldsymbol{v}]$ for the line through the origin generated by $\boldsymbol{v}$. For unit vectors $\boldsymbol{v},\boldsymbol{w}\in\mathbb S^d$, define their projective distance by
\[
d_{\mathrm P}([\boldsymbol{v}],[\boldsymbol{w}])=\arccos|\langle \boldsymbol{v},\boldsymbol{w}\rangle|.
\]

\begin{lemma}\label{lem:packing}
Let $d$ be a positive integer and $[\boldsymbol{v}_1],\ldots,[\boldsymbol{v}_M]$ be $M \ge 2$ distinct lines through the origin in $\mathbb{R}^{d+1}$, each represented by a unit vector $\boldsymbol{v}_i \in \mathbb{S}^{d}$. Then there exist $i\ne j$ such that
\begin{align}\label{eq10}
d_{\mathrm P}([\boldsymbol{v}_i],[\boldsymbol{v}_j])\le \beta_d M^{-1/d}, 
\end{align}
where
\[
\beta_d=\left(\frac{d\pi^{d-1}|\mathbb S^d|}{|\mathbb S^{d-1}|}\right)^{1/d}.
\]
\end{lemma}

\begin{proof}
Let $\delta=\min_{i\ne j} d_{\mathrm P}([\boldsymbol{v}_i],[\boldsymbol{v}_j])$. The $2M$ spherical caps on $\mathbb S^d$ centered at $\pm \boldsymbol{v}_1,\ldots,\pm \boldsymbol{v}_M$ and having angular radius $\frac{\delta}{2}$ have disjoint interiors. 
Here a spherical cap of angular radius $\rho$ centered at $\boldsymbol{s}\in\mathbb S^d$ is the set $\{\boldsymbol{x}\in\mathbb S^d\mid \arccos\langle\boldsymbol{x},\boldsymbol{s}\rangle\le \rho\}$, and $\operatorname{Cap}_d(\rho)$ denotes its surface area. This area does not depend on $\boldsymbol{s}$.
Since $0<\delta\le \frac{\pi}{2}$ and $\sin t\ge \frac{2t}{\pi}$ for $0\le t\le \frac{\pi}{2}$, the area of one such cap satisfies
\begin{align*}
   \operatorname{Cap}_d(\frac{\delta}{2})&=|\mathbb S^{d-1}|\int_0^{\delta/2}\sin^{d-1}t\,dt\\
   &\geq |\mathbb S^{d-1}|\int_0^{\delta/2}\left(\frac{2t}{\pi}\right)^{d-1}dt\\
   &=\frac{|\mathbb S^{d-1}|}{2d\pi^{d-1}}\delta^d.
\end{align*}

Comparing the total area of the $2M$ caps with $|\mathbb S^d|$ gives
\[
2M\cdot \frac{|\mathbb S^{d-1}|}{2d\pi^{d-1}}\delta^d\le |\mathbb S^d|.
\]
Solving for $\delta$ proves Eq. \eqref{eq10}.
\end{proof}

\section{Proof of the Matching Lower Bound}\label{sec:lower}

Based on the preparations in the previous sections, we derive a lower bound for $\Gamma_2(\mathcal{C})$ in this section.

\begin{theorem}\label{thm:main}
Let $2\le r<n$, and let $\mathcal C\subseteq\mathbb R^n$ be a real linear $[n,n-r]$ code. 
Then 
\begin{align}\label{eq12}
\Gamma_2(\mathcal C)\ge \frac{a_r}{\sqrt r\,\beta_{r-1}\,2^{\frac{1}{r-1}}}\, n^{1+\frac{1}{r-1}},
\end{align}
where
\begin{align*}
    a_r=\frac{\Gamma(\frac{r}{2})}{\sqrt{\pi}\,\Gamma(\frac{r+1}{2})}\ \text{and} \  \beta_{r-1}=\left(\frac{d\pi^{r-2}|\mathbb S^{r-1}|}{|\mathbb S^{r-2}|}\right)^{\frac{1}{r-1}}.
\end{align*}

\end{theorem}

\begin{proof}
Let $G=AH=(\boldsymbol{g}_0,\ldots,\boldsymbol{g}_{n-1})\in\mathbb R^{r\times n}$ be a full row rank parity matrix of $\mathcal C$. We may assume that (see Eq. \eqref{eqAAA})
\begin{align*}
B_2^r\subseteq S_G\subseteq \sqrt r\,B_2^r. 
\end{align*}

By Eq. \eqref{eqAAA}, for every unit vector $\boldsymbol{u}\in\mathbb S^{r-1}$,
\begin{align}\label{eqq1}
h_{S_G}(\boldsymbol{u})=\sum_{t=0}^{n-1} |\langle \boldsymbol{u},\boldsymbol{g}_t\rangle|\leq h_{\sqrt{r}B_2^r}(\boldsymbol{u})=\sqrt r.
\end{align}

Let $\boldsymbol{w}=(w_0,w_1,\ldots,w_{r-1})$ be a random unit vector uniformly distributed on $\mathbb S^{r-1}$, and $\mathbb E$ denotes expectation over the random choice of $\boldsymbol{w}$. The uniform distribution on the sphere is invariant under rotations.
For any fixed nonzero $\boldsymbol{g}\in\mathbb R^r$, we have
\begin{align}\label{eqq2}
    \mathbb E[|\langle \boldsymbol{w},\boldsymbol{g}\rangle|]=\|\boldsymbol{g}\|_2\cdot \mathbb E[|\langle \boldsymbol{w},\frac{\boldsymbol{g}}{\|\boldsymbol{g}\|_2}\rangle|]. 
\end{align}
Note that, by symmetry, since $\frac{\boldsymbol{g}}{\|\boldsymbol{g}\|_2} \in \mathbb S^{r-1}$ and $\boldsymbol{w}$ is uniformly distributed on $\mathbb{S}^{r-1}$, we may replace $\frac{\boldsymbol{g}}{\|\boldsymbol{g}\|_2}$ by the unit vector $\boldsymbol{e}_0:=(1,0,\ldots,0)\in \mathbb S^{r-1}$ without changing the expectation, i.e., 
\begin{align}
    \mathbb E[|\langle \boldsymbol{w},\boldsymbol{g}\rangle|]=\|\boldsymbol{g}\|_2\cdot \mathbb E[|\langle \boldsymbol{w},\boldsymbol{e}_0\rangle|]=\|\boldsymbol{g}\|_2\cdot \mathbb E[|w_0|]:=\|\boldsymbol{g}\|_2\cdot a_r,
\end{align}
where (note that the marginal density of a single coordinate of a uniform distribution on \( \mathbb{S}^{r-1} \) is known to be $f_{w_0}(t) = \frac{\Gamma(\frac{r}{2})}{\sqrt{\pi}\,\Gamma(\frac{r-1}{2})} (1 - t^2)^{\frac{r-3}{2}}$ with $|t|\leq 1$)
\begin{align*}
a_r=\mathbb E|w_0|=2 \cdot \frac{\Gamma(\frac{r}{2})}{\sqrt{\pi}\,\Gamma(\frac{r-1}{2})} \int_0^1 t (1 - t^2)^{\frac{r-3}{2}} \, dt=\frac{\Gamma(\frac{r}{2})}{\sqrt{\pi}\,\Gamma(\frac{r+1}{2})}.
\end{align*}

Taking the expectation of Eq. \eqref{eqq1} gives
\begin{align*}
\mathbb E[h_{S_G}(\boldsymbol{w})]=\sum_{t=0}^{n-1}\|\boldsymbol{g}_t\|_2\cdot a_r\leq \sqrt{r}.
\end{align*}
Then we have 
\begin{align}\label{eqq3}
    \sum_{t=0}^{n-1}\|\boldsymbol{g}_t\|_2\leq L_r:=\frac{\sqrt{r}}{a_r}.
\end{align}
Define the set of short-column indices by
\begin{align}\label{eqq4}
J:=\left\{k\in[0,n-1]\mid \|\boldsymbol{g}_k\|_2\le \frac{2L_r}{n}\right\}.
\end{align}
Then $|J|\ge \lceil n/2\rceil$.
Otherwise, the sum in Eq. \eqref{eqq3} would exceed $L_r$.

Because $2\le r<n$, we have $n\ge 3$, so $|J|\ge 2$. 
By Lemma~\ref{lem:packing}, there exist distinct $i,j\in J$ such that the projective angle between $\boldsymbol{g}_i$ and $\boldsymbol{g}_j$,
\[
\theta_{ij}=\arccos\frac{|\langle \boldsymbol{g}_i,\boldsymbol{g}_j\rangle|}{\|\boldsymbol{g}_i\|_2\|\boldsymbol{g}_j\|_2}
\]
satisfies
\begin{align}\label{eqq5}
\theta_{ij}\le \beta_{r-1}\left(\frac{2}{n}\right)^{\frac{1}{r-1}}. 
\end{align}

We now return to the dual expression for $\Gamma_2(\mathcal{C})$ (see Theorem~\ref{prop:dual}).
Consider any vector $\boldsymbol{u}$ feasible for the ordered pair $(i,j)$ in Eq. \eqref{eq4}:
\[
\langle \boldsymbol{u},\boldsymbol{g}_j\rangle=0,\ \langle \boldsymbol{u},\boldsymbol{g}_i\rangle=1.
\]
Since $\langle \boldsymbol{u}, \boldsymbol{g}_j \rangle = 0$, the component of $\boldsymbol{g}_i$ parallel to $\boldsymbol{g}_j$ does not contribute to $\langle \boldsymbol{u}, \boldsymbol{g}_i \rangle$. 
Let $\boldsymbol{v} = \boldsymbol{g}_i - \frac{\langle \boldsymbol{g}_i, \boldsymbol{g}_j \rangle}{\|\boldsymbol{g}_j\|_2^2} \boldsymbol{g}_j$ be the orthogonal projection of $\boldsymbol{g}_i$ onto $\boldsymbol{g}_j^\perp$. 
Then $\langle \boldsymbol{u}, \boldsymbol{g}_i \rangle = \langle \boldsymbol{u}, \boldsymbol{v} \rangle$, and $\|\boldsymbol{v}\|_2 = \|\boldsymbol{g}_i\|_2 \sin (\theta_{ij})$. 
Hence,
\begin{align*}
    1 = |\langle \boldsymbol{u}, \boldsymbol{v} \rangle| \le \|\boldsymbol{u}\|_2 \|\boldsymbol{v}\|_2 = \|\boldsymbol{u}\|_2 \|\boldsymbol{g}_i\|_2 \sin (\theta_{ij}),
\end{align*}
which gives 
\begin{align}\label{eqq6}
    \|\boldsymbol{u}\|_2 \ge \frac{1}{\|\boldsymbol{g}_i\|_2 \sin (\theta_{ij})}.
\end{align}

The lower inclusion in Eq. \eqref{eqAAA} gives
\[
h_{S_G}(\boldsymbol{u})\ge h_{B_2^r}(\boldsymbol{u})=\|\boldsymbol{u}\|_2.
\]
Using Theorem~\ref{prop:dual}, Eqs.~\eqref{eqq4}, \eqref{eqq5}, and \eqref{eqq6}, and $\sin(\theta)\le \theta$ for all $\theta\geq 0$, we obtain
\begin{align*}
    \Gamma_2(\mathcal C)
&\ge 2\min_{\substack{\langle \boldsymbol{u},\boldsymbol{g}_j\rangle=0\\ \langle \boldsymbol{u},\boldsymbol{g}_i\rangle=1}} h_{S_G}(\boldsymbol{u})\\
&\ge \frac{2}{\|\boldsymbol{g}_i\|_2\sin\theta_{ij}}\\
&\ge \frac{2}{(\frac{2L_r}{n})\cdot \beta_{r-1}(\frac{2}{n})^{\frac{1}{r-1}}}
\\
&=\frac{a_r}{\sqrt r\,\beta_{r-1}\,2^{\frac{1}{r-1}}}\, n^{1+\frac{1}{r-1}}.
\end{align*}
Therefore, the Theorem is proved.
\end{proof}

Next, we derive a simpler, uniform lower bound.

\begin{lemma}\label{lem:constants}
For every integer $r\ge 2$ and $d\ge 1$,
\[
a_r\ge \frac{1}{\sqrt{3r}},\  \beta_d\le 2\pi.
\]
\end{lemma}

\begin{proof}
Let $X=|u_0|$ for $\boldsymbol{u}=(u_0,u_1,\ldots,u_{r-1})$ uniform on $\mathbb S^{r-1}$. By symmetry, 
\begin{align*}
    \mathbb E[X^2] =\frac{1}{r}\cdot(\sum_{i=0}^{r-1} \mathbb E[u_i^2])=\frac{1}{r}\cdot(\mathbb E[\sum_{i=0}^{r-1} u_i^2])=\frac{1}{r}.
\end{align*}
Let
\[
A=\mathbb E[u_0^4] ,\ B=\mathbb E[u_0^2 u_1^2].
\]
Actually, by symmetry of the spherical distribution, \( A = \mathbb{E}[u_i^4] \) for all \( i \in[0,r-1]\), and \( B = \mathbb{E}[u_i^2 u_j^2] \) for all \( i,j\in[0,r-1]\) with \(i \neq j \).
Consider the orthogonal transformation \( R\in\mathbb R^{r\times r} \) that rotates the first two coordinates by \( 45^\circ \):
\[
R = \begin{pmatrix}
\frac{1}{\sqrt{2}} & \frac{1}{\sqrt{2}} & 0 & \cdots & 0 \\
-\frac{1}{\sqrt{2}} & \frac{1}{\sqrt{2}} & 0 & \cdots & 0 \\
0 & 0 & 1 & \cdots & 0 \\
\vdots & \vdots & \vdots & \ddots & \vdots \\
0 & 0 & 0 & \cdots & 1
\end{pmatrix}.
\]
Let $\boldsymbol{v}=R\boldsymbol{u}=(v_0,v_1,\ldots,v_{r-1})$.
Since the uniform distribution on \( \mathbb{S}^{r-1} \) is rotationally invariant, \( \boldsymbol{u} \) and \( \boldsymbol{v} \) have the same distribution.
In particular, the first coordinates satisfy
\[
\mathbb{E}[u_0^4] = \mathbb{E}[v_0^4].
\]
Note that
$v_0 = \frac{u_0 + u_1}{\sqrt{2}}.$
Thus
\[
A = \mathbb{E}\left[ \left( \frac{u_0 + u_1}{\sqrt{2}} \right)^4 \right]
= \frac{1}{4} \mathbb{E}\left[ (u_0 + u_1)^4 \right].
\]
Furthermore, we have
\[
A = \frac{1}{4} \mathbb{E}\left[ u_0^4 + 4u_0^3 u_1 + 6u_0^2 u_1^2 + 4u_0 u_1^3 + u_1^4 \right].
\]
Because the distribution is symmetric under the sign change \( u_i \mapsto -u_i \), all moments containing an odd power of \( u_0 \) or \( u_1 \) vanish. 
Hence
\[
\mathbb{E}[u_0^3 u_1] = \mathbb{E}[u_0 u_1^3] = 0.
\]
Therefore
\[
A = \frac{1}{4} \left( \mathbb{E}[u_0^4] + 6\mathbb{E}[u_0^2 u_1^2] + \mathbb{E}[u_1^4] \right).
\]
Using symmetry to replace \( \mathbb{E}[u_0^4] = \mathbb{E}[u_1^4] = A \) and \( \mathbb{E}[u_0^2 u_1^2] = B \), we get
\[
A = \frac{1}{4} (A + 6B + A).
\]
Then we have $A=3B$.

On the other hand, since 
\(\left( \sum_{i=0}^{r-1} u_i^2 \right)^2 = 1.\)
Expanding the square gives
\[
\sum_{i=0}^{r-1} u_i^4 + \sum_{0\le i \neq j\le r-1} u_i^2 u_j^2 = 1.
\]
Taking expectations on both sides, we have
\[
\mathbb{E}\left[ \sum_{i=0}^{r-1} u_i^4 \right] + \mathbb{E}\left[ \sum_{0\le i \neq j\le r-1} u_i^2 u_j^2 \right] = 1.
\]
By symmetry,
\[
\mathbb{E}\left[ \sum_{i=0}^{r-1} u_i^4 \right] = rA,  \]
and 
\[ \mathbb{E}\left[ \sum_{0\le i \neq j\le r-1} u_i^2 u_j^2 \right] = r(r-1)B.\] 
Thus
\[
rA + r(r-1)B = 1.
\]
Combining with $A=3B$, we have
\(A=\frac{3}{r(r+2)}.
\)
According to H\"older's inequality, we have
\[
\mathbb E [X^2]=\mathbb E[X^{2/3}X^{4/3}]\le (\mathbb E [X])^{2/3}(\mathbb E [X^4])^{1/3},
\]
which yields
\[
a_r=\mathbb E [X]\ge \frac{(\mathbb E [X^2])^{3/2}}{(\mathbb E [X^4])^{1/2}}
=\frac{\sqrt{r+2}}{\sqrt{3}\,r}\ge \frac{1}{\sqrt{3r}}.
\]

For the bound of $\beta_{d}$, the spherical-coordinate formula gives
\[
\frac{|\mathbb S^d|}{|\mathbb S^{d-1}|}=\int_0^\pi \sin^{d-1}t\,dt\le \pi.
\]
Hence
\[
\beta_d\le \pi d^{1/d}\le 2\pi,
\]
where the last inequality follows from $d\le 2^d$.
\end{proof}

\begin{corollary}\label{cor:v3}
For every real $[n,n-r]$ linear code $\mathcal{C}$ with $2\le r<n$,
\begin{align}\label{eqqls}
\Gamma_2(\mathcal C)\ge \frac{1}{4\pi\sqrt{3}r}\, n^{1+\frac{1}{r-1}}.
\end{align}
\end{corollary}

\begin{proof}
Combine Lemma~\ref{lem:constants} with $2^{\frac{1}{r-1}}\le 2$, we have
\[
\frac{a_r}{\sqrt r\,\beta_{r-1}\,2^{\frac{1}{r-1}}}
\ge \frac{1}{4\pi\sqrt{3}r}.
\]
Thus we have Eq. \eqref{eqqls}.
\end{proof}

\textbf{Remark.}
Combining the construction in~\cite{AECC202606} with Corollary~\ref{cor:v3}, it follows that for any fixed redundancy $r\geq 3$, the leading exponent $n^{1+\frac{1}{r-1}}$ in the lower bound on $\Gamma_2(\mathcal C)$ is optimal. 
This provides a criterion for assessing whether the leading term of the error-correction capability of existing explicit single-error-correcting Analog ECC constructions is optimal (at least up to a constant factor). The corresponding results are summarized in Table~\ref{tab:summary}.
For the case of $r=2$, the $C_0$ code in~\cite{AECC2020} has been proven to have optimal (and is tight) threshold, so we do not list it in the table.

\begin{table}[htbp]
\centering
\caption{Results on the Optimality of the Leading Term of the Single-Error Correction Threshold for Existing Analog ECC Constructions.}
\label{tab:summary}
\small
\renewcommand{\arraystretch}{1.25}
\begin{tabularx}{\textwidth}{@{}>{\centering\arraybackslash}m{2.8cm}|>{\centering\arraybackslash}m{3.8cm}|>{\centering\arraybackslash}X|
>{\centering\arraybackslash}m{2.0cm}@{}}
\hline
\hline
Construction &Redundancy $r$ &The single-error correction threshold $\Gamma_2(\mathcal{C})$ & Whether the leading term is optimal \\
\hline
Construction in \cite{2025Analog}& 2 & $O(n^2)$& Yes \\
\hline
Construction in \cite{song2026AECC}&3 & $\frac{2n}{\sin(\frac{\pi}{2\lceil \sqrt{\frac{n-1}{2}}\rceil})}\ \sim\ O(n\sqrt{n})$& Yes \\
\hline
Proposition 6 in \cite{AECC2020} &$r^2-r\geq n\geq r$ (here we take $r=O(\sqrt{n})$)& $2\ \lceil \frac{2n}{r}\rceil\sim\ O(\sqrt{n})$& Yes \\
\hline
Proposition 5 in \cite{9965851} &$\Theta(\log(n))$& $ O(\frac{n}{\sqrt{\log(n)}})$& No \\
\hline
Corollary 16 in \cite{AECC20242} &$(\ell +1)q\sim \Theta(n^{\frac{1}{\ell +1}})$ ($n=q^{\ell+1}$ and $q$ is a prime power)& $\frac{2(\ell +1)n}{r}\sim O(n^{\frac{\ell}{\ell +1}})$& Yes \\
\hline
Construction $\mathcal{C}_1$ in \cite{AECC202606} &Fixed $r\geq 2$& $4n\left\lceil 
  \frac{n^{\frac{1}{r-1}} \Gamma\!\left(\frac{r+1}{2}\right)^{\frac{1}{r-1}}}{\sqrt{\pi}} 
  + \sqrt{\frac{r-1}{2}} 
\right\rceil\sim O(n^{1+\frac{1}{r-1}})$& Yes \\
\hline
Construction in \cite{zhu2026new} &$n-2$& $2\frac{\cos(\frac{\pi}{2n})}{\cos(\frac{3\pi}{2n})}+2\sim O(1)$& Yes \\
\hline
\hline
\end{tabularx}
\end{table}

\section{Conclusion}
\label{sec:con}

In this paper, we have established the first asymptotic lower bound for the single-error correction threshold $\Gamma_2(\mathcal{C})$ of analog error-correcting codes. For every real linear $[n, k=n-r]$ code $\mathcal{C}$ with $2 \le r < n$, we prove that $\Gamma_2(\mathcal{C}) \ge \frac{1}{4\pi\sqrt{3}r} \cdot n^{1+\frac{1}{r-1}}$, which, together with the matching upper bound from \cite{AECC202606}, confirms the optimality of the exponent $n^{1+\frac{1}{r-1}}$ for every fixed $r \ge 2$. 
One of our future directions is to establish optimality bounds for the error-correction threshold $\Gamma_m(\mathcal{C})$ in the multi-error scenario with $m \ge 3$.

\ifCLASSOPTIONcaptionsoff
  \newpage
\fi

\bibliographystyle{IEEEtran}
\bibliography{CNC-v1}

@String { ISIT         = {Proc. {IEEE} Int. Symp. Inf. Theory} }

@misc{song2026AECC,
      title={Analog Error Correcting Codes with Constant Redundancy}, 
      author={Wentu Song and Kui Cai},
      year={2026},
      eprint={2603.07117},
      archivePrefix={arXiv},
      primaryClass={cs.IT},
      url={https://arxiv.org/abs/2603.07117}, 
}

@misc{AECC202605,
      title={Tight Lower Bounds on The Single-Error Detection Threshold for Analog Error-Correcting Codes}, 
      author={Zhengyi Jiang and Wenhao Liu and Zhongyi Huang and Bo Bai and Gong Zhang and Hanxu Hou},
      year={2026},
      eprint={2605.08973},
      archivePrefix={arXiv},
      primaryClass={cs.IT},
      url={https://arxiv.org/abs/2605.08973}, 
}

@misc{AECC202606,
      title={Sharp Bounds and New Constructions for Single-Error Detection and Correction in Analog Codes}, 
      author={Hengzhuo Li and Zhengjie Jian and Xin Wang and Hengjia Wei},
      year={2026},
      eprint={2606.03011},
      archivePrefix={arXiv},
      primaryClass={cs.IT},
      url={https://arxiv.org/abs/2606.03011}, 
}

@misc{roth2026height,
      title={On the Height Profile of Analog Error-Correcting Codes}, 
      author={Ron M. Roth and Ziyuan Zhu and Changcheng Yuan and Paul H. Siegel and Anxiao Jiang},
      year={2026},
      eprint={2602.20366},
      archivePrefix={arXiv},
      primaryClass={cs.IT},
      url={https://arxiv.org/abs/2602.20366}, 
}

@INPROCEEDINGS{9965851,
  author={Roth, Ron M.},
  booktitle={2022 IEEE Information Theory Workshop (ITW)}, 
  title={Fault-Tolerant Neuromorphic Computing on Nanoscale Crossbar Architectures}, 
  year={2022},
  volume={},
  number={},
  pages={202-207},
  doi={10.1109/ITW54588.2022.9965851}}

@book{boyd2004convex,
  title={Convex Optimization},
  author={Boyd, Stephen and Vandenberghe, Lieven},
  year={2004},
  publisher={Cambridge University Press}
}

@INPROCEEDINGS{2025Analog,
  author={Jiang, Zhengyi and Shi, Hao and Huang, Zhongyi and Bai, Bo and Zhang, Gong and Hou, Hanxu},
  booktitle={2025 IEEE International Symposium on Information Theory (ISIT)}, 
  title={Constructions of Analog Error-Correcting Codes for Single-Error Detection and Correction with Efficient Decoding Algorithm}, 
  year={2025},
  volume={},
  number={},
  pages={1-6},
  doi={10.1109/ISIT63088.2025.11195392}}

@article{zhu2026new,
  title={A New Class of Geometric Analog Error Correction Codes for Crossbar Based In-Memory Computing},
  author={Zhu, Ziyuan and Yuan, Changcheng and Roth, Ron M and Siegel, Paul H and Jiang, Anxiao},
  journal={arXiv preprint arXiv:2603.03723},
  year={2026}
}

@article{zhang2023edge,
  title={Edge learning using a fully integrated neuro-inspired memristor chip},
  author={Zhang, Wenbin and Yao, Peng and Gao, Bin and Liu, Qi and Wu, Dong and Zhang, Qingtian and Li, Yuankun and Qin, Qi and Li, Jiaming and Zhu, Zhenhua and others},
  journal={Science},
  volume={381},
  number={6663},
  pages={1205--1211},
  year={2023},
  publisher={American Association for the Advancement of Science}
}

@article{sebastian2020memory,
  title={Memory devices and applications for in-memory computing},
  author={Sebastian, Abu and Le Gallo, Manuel and Khaddam-Aljameh, Riduan and Eleftheriou, Evangelos},
  journal={Nature nanotechnology},
  volume={15},
  number={7},
  pages={529--544},
  year={2020},
  publisher={Nature Publishing Group UK London}
}

@ARTICLE{50305,
  author={Kub, F.J. and Moon, K.K. and Mack, I.A. and Long, F.M.},
  journal={IEEE Journal of Solid-State Circuits}, 
  title={Programmable analog vector-matrix multipliers}, 
  year={1990},
  volume={25},
  number={1},
  pages={207-214},
  doi={10.1109/4.50305}}

@inproceedings{hu2016dot,
  title={Dot-product engine for neuromorphic computing: Programming 1T1M crossbar to accelerate matrix-vector multiplication},
  author={Hu, Miao and Strachan, John Paul and Li, Zhiyong and Grafals, Emmanuelle M and Davila, Noraica and Graves, Catherine and Lam, Sity and Ge, Ning and Yang, Jianhua Joshua and Williams, R Stanley},
  booktitle={Proceedings of the 53rd annual design automation conference},
  pages={1--6},
  year={2016}
}

@ARTICLE{104196,
  author={Boser, B.E. and Sackinger, E. and Bromley, J. and Le Cun, Y. and Jackel, L.D.},
  journal={IEEE Journal of Solid-State Circuits}, 
  title={An analog neural network processor with programmable topology}, 
  year={1991},
  volume={26},
  number={12},
  pages={2017-2025},
  doi={10.1109/4.104196}}

@article{AECC20242,
author={Jiang, Anxiao},
  journal={IEEE Transactions on Information Theory}, 
  title={Analog Error-Correcting Codes: Designs and Analysis}, 
  year={2024},
  volume={70},
  number={11},
  pages={7740-7756},
  doi={10.1109/TIT.2024.3454059}}

@ARTICLE{AECC2024,
  author={Wei, Hengjia and Roth, Ron M.},
  journal={IEEE Transactions on Information Theory}, 
  title={Multiple-Error-Correcting Codes for Analog Computing on Resistive Crossbars}, 
  year={2024},
  volume={70},
  number={12},
  pages={8647-8658},
  doi={10.1109/TIT.2024.3439674}}

@ARTICLE{AECC2020,
  author={Roth, Ron M.},
  journal={IEEE Transactions on Information Theory}, 
  title={Analog Error-Correcting Codes}, 
  year={2020},
  volume={66},
  number={7},
  pages={4075-4088},
  doi={10.1109/TIT.2020.2977918}}

@INPROCEEDINGS{AECC2019,
  author={Roth, Ron M.},
  booktitle={2019 IEEE International Symposium on Information Theory (ISIT)}, 
  title={Analog Error-Correcting Codes}, 
  year={2019},
  volume={},
  number={},
  pages={2419-2423},
  doi={10.1109/ISIT.2019.8849843}}

@book{Goodfellow-et-al-2016,
    title={Deep Learning},
    author={Ian Goodfellow and Yoshua Bengio and Aaron Courville},
    publisher={MIT Press},
    note={\url{http://www.deeplearningbook.org}},
    year={2016}
}

@InCollection{ball1997,
  author    = {Keith M. Ball},
  title     = {An Elementary Introduction to Modern Convex Geometry},
  booktitle = {Flavors of Geometry},
  editor    = {Silvio Levy},
  series    = {Mathematical Sciences Research Institute Publications},
  volume    = {31},
  pages     = {1--58},
  publisher = {Cambridge University Press},
  address   = {Cambridge},
  year      = {1997}
}

\end{document}